\documentclass[runningheads]{llncs}

\usepackage[utf8]{inputenc}

\usepackage{algorithm}
\usepackage[noend]{algpseudocode}
\usepackage{amsmath}
\usepackage{amssymb}
\usepackage{booktabs}
\usepackage{enumitem}
\usepackage{mathpartir}
\usepackage{microtype}
\usepackage{relsize}
\usepackage{stmaryrd}
\usepackage{thm-restate}

\usepackage{tikz}
\usetikzlibrary{positioning}

\usepackage[normalem]{ulem}
\usepackage{todonotes}

\newcommand{\ML}[1]{\todo[inline]{#1}}
\newcommand{\MLadd}[1]{{\color{red} #1}}
\newcommand{\MLdel}[1]{\MLadd{\sout{#1}}}

\newcommand{\TW}[1]{\todo[inline]{#1}}

\newcommand{\TWdel}[1]{\MLadd{\sout{#1}}}

\renewcommand{\ML}[1]{}
\renewcommand{\MLadd}[1]{#1}
\renewcommand{\MLdel}[1]{}

\renewcommand{\TW}[1]{}

\renewcommand{\TWdel}[1]{}

\setlist[itemize]{noitemsep, topsep=0pt}
\setlist[enumerate]{noitemsep, topsep=0pt}

\newcommand{\eg}{\emph{e.g.}}
\newcommand{\ie}{\emph{i.e.}}

\newcommand{\lscott}{\ensuremath{[\![}}
\newcommand{\rscott}{\ensuremath{]\!]}}

\newcommand{\multiset}[1]{\Lbag #1 \Rbag}
\newcommand{\set}[1]{{\{ #1 \}}}
\newcommand{\sem}[1]{\mathbf{#1}}
\DeclareMathOperator*{\bigplus}{\mathlarger{\mathlarger{\mathlarger{+}}}}

\newcommand{\booleans}{\mathbb{B}}
\newcommand{\naturalnumbers}{\mathbb{N}}

\newcommand{\subst}{\sigma}

\newcommand{\interpret}[1]{\ensuremath{\lscott #1\mkern1mu\rscott}}
\newcommand{\multiact}{\alpha}
\newcommand{\nodata}[1]{\underline{#1}}
\newcommand{\sequential}{\mathbin{.}}

\newcommand{\allow}{\nabla}
\newcommand{\communication}{\Gamma}
\newcommand{\comms}{\mathsf{Comm}}
\newcommand{\hide}{\tau}

\newcommand{\actcomm}[1]{\gamma_{#1}}
\newcommand{\acthide}[1]{\theta_{#1}}
\newcommand{\freevars}{\textsf{FV}}

\newcommand{\lts}{\mathcal{L}}
\newcommand{\states}{\mathit{S}}
\newcommand{\events}{\Lambda}

\newcommand{\actions}{\mathit{Act}}
\newcommand{\semactions}{\Omega}
\newcommand{\action}{\omega}

\newcommand{\transitions}{\mathbin{\rightarrow}}
\newcommand{\transition}[1]{\xrightarrow{#1}}

\newcommand{\simpleprocess}{\mathsf{S}}
\newcommand{\linearprocess}{\mathsf{P}}
\newcommand{\names}{PN}
\newcommand{\initval}{\iota}

\newcommand{\actsync}{\mathsf{sync}}
\newcommand{\acttag}{\mathsf{tag}}
\newcommand{\Iind}{K}
\newcommand{\Isub}{J}

\newcommand{\REQ}[1]{\ifx1#1\textsc{SYN}\else\ifx2#1\textsc{IND}\else\ifx3#1\textsc{ORG}\else\ifx4#1\textsc{COM}\fi\fi\fi\fi}

\newcommand{\bisim}{\mathbin{\leftrightarroweq}}

\newcommand{\related}{\mathbin{R}}

\newcommand{\boolsort}{B}
\newcommand{\natsort}{N}

\newcommand{\false}{\textsf{false}}
\newcommand{\true}{\textsf{true}}

\newcommand{\actcount}{\textsf{count}}

\newcommand{\acttoggle}{\textsf{toggle}}

\newcommand{\procmachine}{\text{Machine}}

\newcommand{\subvector}[2]{{#1}_{|#2}}
\newcommand{\flatten}[1]{\mathsf{Vars}(#1)}
\renewcommand{\vector}[1]{\langle #1 \rangle}

\newcommand{\chatbox}{\textsf{Chatbox}}
\newcommand{\hesselink}{\textsf{Register}}
\newcommand{\WMS}{\textsf{WMS}}
\newcommand{\abp}{\textsf{ABP}}

\author{Maurice Laveaux \and Tim A.C.~Willemse}

\title{Decompositional Minimisation of Monolithic Processes}

\institute{Eindhoven University of Technology, Eindhoven, The Netherlands \\ \texttt{\{m.laveaux, t.a.c.willemse\}@tue.nl}}
\date{}

\begin{document}
\maketitle

\begin{abstract}
	Compositional minimisation can be an effective technique to reduce the state space explosion problem.
  This technique considers a parallel composition of several processes.
  In its simplest form, each sequential process is replaced by an \emph{abstraction}, simpler than the corresponding process while still preserving the property that is checked.
  However, this technique cannot be applied in a setting where parallel composition is first translated to a non-deterministic sequential monolithic process.
  The advantage of this monolithic process is that it facilitates static analysis of global behaviour.
  Therefore, we present a technique that considers a monolithic process with data and decomposes it into two processes where each process defines behaviour for a subset of the parameters of the monolithic process.
  We prove that these processes preserve the properties of the monolithic process under a suitable synchronisation context.
  Moreover, we prove that state invariants can be used to improve its effectiveness.
  Finally, we apply the decomposition technique to several specifications.
\end{abstract}

\section{Introduction}

The mCRL2 language~\cite{GrooteM2014} is a process algebra that can be used to specify the behaviour of communicating processes with data.
The corresponding mCRL2 toolset~\cite{BunteGKLNVWWW19} translates the parallel composition and action synchronisation present in the process specification to an equivalent (non-deterministic sequential) monolithic process.
The advantages of this translation are that the design of further static analysis techniques and the implementation of state space exploration can be greatly simplified.
However, the static analysis techniques available at the moment are not always strong enough to mitigate the state space explosion problem for this monolithic process even though its state space can often be minimised modulo some equivalence relation after state space exploration.

In the literature there are several promising techniques to reduce this problem and one of them is \emph{compositional minimisation}.
The general idea is that the state space explosion often occurs due to all the possible interleaving of several processes in a parallel composition.
To reduce the interleaving, in compositional minimisation the state space of each sequential process, referred to as a component, is replaced by an \emph{abstraction}, simpler than the corresponding state space such that their composition preserves the property that is checked~\cite{TaiK93:incremental,TaiK93:hierarchy}.

This naive approach is not always useful, because the size of the state spaces belonging to individual components summed together might exceed the size of the whole state space~\cite{GaravelLM18}.
In particular, the state space can become infinitely large for components that rely on synchronisation to bound their behaviour.
This can be avoided by specifying or generating \emph{interface constraints} (also known as \emph{environmental constraints} or \emph{context constraints}) leading to a \emph{semantic compositional minimisation}~\cite{GrafSL96,CheungK96}.
Furthermore, the order in which intermediate components are explored, minimised and subsequently composed heavily influences the size of the intermediate state spaces.
There are heuristics for these problems that can be very effective in practice, as shown by the CADP~\cite{GaravelLMS13} (Construction and Analysis of Distributed Processes) toolset.

Unfortunately, in our context where parallel composition is removed by a translation step the aforementioned compositional techniques, which rely on the user-defined parallel composition and the (sequential) processes, are not applicable.
In this paper we define a decomposition technique of a monolithic process based on a partitioning of its data parameters that results in two components, which we refer to as a \emph{cleave}, and show that it is correct.
Furthermore, we show that \emph{state invariants} can be used to retain more global information in both components to improve its effectiveness.
Finally, we perform a case study to evaluate the decomposition technique and the advantage of state invariants in practice.

The advantages of decomposing the monolithic process are the same as for compositional minimisation; in that by minimising the state spaces of intermediate components the composed state space can be immediately smaller than the state space obtained by exploring the monolithic process.
Indeed, the case studies on which we report support both observations.
Furthermore, state space exploration relies on the evaluation of data expressions of the higher-level specification language and that can be costly, whereas, composition can be computed without evaluating expressions.
An advantage of the decomposition technique over compositional minimisation is that constraints resulting from the synchronisation of these processes can be used when deriving the components.
Another advantage is that the components resulting from the decomposition are not restricted to the user-defined processes present in the specification, which could yield a more optimal composition.

\paragraph{Related Work.}
Several different techniques are related to this type of decomposition.
Most notably, the work on decomposing of petri nets into a set of automata~\cite{BouvierGL20} also aims to speed up state space exploration by means of decomposition.
The work on functional decomposition~\cite{BrinksmaLB93} describes a technique to decompose a specification based on a partitioning of the action labels instead of a partitioning of the data parameters.
In~\cite{JongmansCP16} it was shown how this type of decomposition can be achieved for mCRL2 processes.
Furthermore, a decomposition technique was used in~\cite{GrooteM92} to improve the efficiency of equivalence checking.
However, that work considers processes that were already in a parallel composition and further decomposes them based on the actions that occur in each component.

\paragraph{Outline.} 
In Section~\ref{section:preliminaries} the syntax and semantics of the considered process algebra are defined.
The decomposition problem is defined in Section~\ref{section:decomposition} and the cleave technique is presented.
In Section~\ref{section:state_invariant} the cleave technique is improved with state invariants.
In Section~\ref{section:casestudy} a case study is presented to illustrate the effectiveness of the decomposition technique in practice.
Finally, a conclusion and future work is presented in Section~\ref{section:conclusion}.

\section{Preliminaries}\label{section:preliminaries}

We assume the existence of an abstract data theory that describes data sorts.
Each sort $D$ has an associated non-empty semantic domain denoted by $\mathbb{D}$.
The existence of sorts $\boolsort$ and $\natsort$ with their associated Boolean ($\booleans$) and natural number ($\naturalnumbers$) semantic domains respectively, with standard operators is assumed.
Furthermore, we assume the existence of an infinite set of \emph{sorted variables}.
We use $e : D$ to indicate that $e$ is an expression (or variable) of sort $D$.
We use $\freevars(e)$ to denote the set of free variables of an expression $e$.
A variable that is not free is called \emph{bound}.
An expression $e$ is \emph{closed} iff $\freevars(e) = \emptyset$.

We use an \emph{interpretation} function, denoted by $\interpret{\ldots}$, which maps syntactic objects to values within their corresponding semantic domain.
We assume that for any closed expression $e$ of a sort described by the data theory that $\interpret{e}$ is already defined.
We typically use boldface for semantic objects to differentiate them from syntax, \eg, the semantic object associated with the expression $1 + 1$ is $\textbf{2}$.
For most operators we use the same symbol in both syntactic and semantic domains.
However, we denote \emph{data equivalence} by $e \approx f$, which is $\true$ iff $\interpret{e} = \interpret{f}$.

We use the following notation for vectors.
Given a \emph{vector} $\vec{d} = \vector{d_0, \ldots, d_n}$ of length $n + 1$.
Two vectors are equivalent, denoted by $\vector{d_0, \ldots, d_n} \approx \vector{e_0, \ldots, e_n}$, iff their elements are \emph{pairwise equivalent}, \ie, $d_i \approx e_i$ for all $0 \leq i \leq n$.
Given a vector $\vector{d_0, \ldots, d_n}$ and a subset $I \subseteq \naturalnumbers$, we define the \emph{projection}, denoted by $\subvector{\vector{d_0, \ldots, d_n}}{I}$, as the vector $\vector{d_{i_0}, \ldots, d_{i_l}}$ for the largest $l \in \naturalnumbers$ such that $i_0 < i_1 < \ldots < i_l \leq n$ and $i_k \in I$ for $0 \leq k \leq l$.
We write $\vec{d} : \vec{D}$ for a vector of $n + 1$ variables $d_0 : D_0, \ldots, d_n : D_n$ and denote the projection for a subset of indices $I \subseteq \naturalnumbers$ by $\subvector{\vec{d}}{I} : \subvector{\vec{D}}{I}$.
Finally, we define $\flatten{\vec{d}} = \{d_0, \ldots, d_n\}$.

Given a set $A$ we consider any total function $A \rightarrow \naturalnumbers$ to denote a \emph{multi-set}.
Let $m, m' : A \rightarrow \naturalnumbers$.
\emph{Inclusion}, denoted by $m \subseteq m'$, is defined as $m \subseteq m'$ iff for all elements $a \in A$ it holds that $m(a) \leq m'(a)$.
Furthermore, we define the binary operator $m + m'$ as the pointwise addition: $\forall a \in A : (m + m')(a) = m(a) + m'(a)$. 
Similarly, we define the dual operator $m - m'$ such that $\forall a \in A : (m - m')(a) = \text{max}(m(a) - m'(a), 0)$.
Given an element $a \in A$ we refer to $m(a)$ as the \emph{multiplicity} of $a$.
We use the notation $\multiset{\ldots}$ for a multi-set where the multiplicity of each element is either written next to it or omitted when it is one, \eg, $\multiset{a : 2, b}$ has elements $a$ and $b$ with multiplicity two and one respectively (and all other elements have multiplicity zero).

\subsection{Labelled Transition Systems}

Let $\events$ be the set of (sorted) action labels.
We use $D_a$ to indicate the sort of action label $a \in \events$.

\begin{definition}
  \emph{Multi-actions} are defined as follows:
  \begin{equation*}
    \multiact ~::=~ \tau ~\mid~ a(e) ~\mid~ \multiact | \multiact
  \end{equation*}
  Where $\tau$ denotes the \emph{invisible} action and $a \in \events$ is an action label with an expression $e$ of sort $D_a$.
  Finally, $\multiact | \multiact$ denotes the simultaneous occurrence of two (multi-)actions.
  
  The set of all multi-sets over $\set{a(\sem{e}) \mid a \in \events, \mathbf{e} \in \mathbb{D}_a}$ is denoted $\semactions$.
 	The semantics of a multi-action $\multiact$, denoted by $\interpret{\multiact}$, is an element of $\semactions$ and defined inductively as follows: $\interpret{\tau} = \multiset{}$, $\interpret{a(e)} = \multiset{a(\interpret{e})}$ and $\interpret{\alpha | \beta} = \interpret{\alpha} + \interpret{\beta}$.
\end{definition}

We assume the existence of a singleton sort $\bot$ and omit the parenthesis and expression of action labels of sort $\bot$ to provide a uniform treatment of actions with data and actions without relevant data parameters.

\begin{definition}
  A labelled transition system with multi-actions, abbreviated LTS, is a tuple $\lts = (\states, s_0, \actions, \transitions)$ where $\states$ is a set of states; $s_0 \in S$ is an initial state; $\actions \subseteq \semactions$ and $\transitions\,\subseteq \states \times \actions \times \states$ is a labelled transition relation.
\end{definition}

We typically use $\action$ to denote an element of $\semactions$ and we write $s \transition{\action} t$ whenever $(s, \action, t) \in \transitions$.
We depict LTSs as edge-labelled directed graphs, where vertices represent states and the labelled edges between vertices represent the transitions. 
An incoming arrow with no starting state and no multi-action indicates the initial state.
An example of an LTS that models the behaviour of a machine is depicted in Figure~\ref{figure:delayedmachine}.

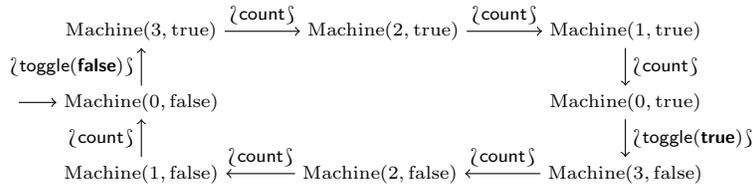
\begin{figure}
	\begin{center}
	\begin{tikzpicture}[->]
 		\scriptsize
	  \node (s0) {$\procmachine(0, \text{false})$};
	  \node (init) [left=5mm of s0] {};
	  \draw (init) edge (s0);
	      
	  \node (s1) [above=.5cm of s0] {$\procmachine(3, \text{true})$};
	  \node (s2) [right=of s1] {$\procmachine(2, \text{true})$};
	  \node (s3) [right=of s2] {$\procmachine(1, \text{true})$};
	  \node (s4) [below=.5cm of s3] {$\procmachine(0, \text{true})$};
	  \node (s5) [below=.5cm of s4] {$\procmachine(3, \text{false})$};
	  \node (s6) [left=of s5] {$\procmachine(2, \text{false})$};
	  \node (s7) [below=.5cm of s0] {$\procmachine(1, \text{false})$};

	  \draw (s0) edge node[left] {$\multiset{\acttoggle(\textbf{\false})}$} (s1);	  
	  \draw (s1) edge node[above] {$\multiset{\actcount}$} (s2);
	  \draw (s2) edge node[above] {$\multiset{\actcount}$} (s3);
	  \draw (s3) edge node[right] {$\multiset{\actcount}$} (s4);
	  \draw (s4) edge node[right] {$\multiset{\acttoggle(\textbf{\true})}$} (s5);
	  \draw (s5) edge node[above] {$\multiset{\actcount}$} (s6);
	  \draw (s6) edge node[above] {$\multiset{\actcount}$} (s7);
	  \draw (s7) edge node[left] {$\multiset{\actcount}$} (s0);  
	\end{tikzpicture}	
	\end{center}
  \caption{Example LTS for the behaviour of a machine.}\label{figure:delayedmachine}
\end{figure}

We recall the well-known strong bisimulation equivalence relation on LTSs~\cite{Milner83}.

\begin{definition}
  Let $\lts_i = (\states_i, s_i, \actions_i, \transitions_i)$ for $i \in \{1,2\}$ be two LTSs.
  A binary relation $R \subseteq \states_1 \times \states_2$ is a \emph{strong bisimulation relation} iff for all $s \related t$:
  \begin{itemize}
 		\item if $s \transition{\action}_1 s'$ then there is a state $t' \in \states_2$ such that $t \transition{\action}_2 t'$ and $s' \related t'$, and
 		
 		\item if $t \transition{\action}_2 t'$ then there is a state $s' \in \states_1$ such that $s \transition{\action}_1 s'$ and $s' \related t'$.
  \end{itemize}

  Two states $s$ and $t$ are \emph{strongly bisimilar}, denoted by $s \bisim t$, iff there is a strong bisimulation relation $R$ such that $s \related t$.
  LTSs $\lts_1$ and $\lts_2$ are strongly bisimilar, denoted by $\lts_1 \bisim \lts_2$, iff $s_1 \bisim s_2$.  
\end{definition}

\subsection{Linear Process Equations}

The (non-deterministic sequential) monolithic processes that we consider are defined by a number of \emph{condition-action-effect} statements; which are called \emph{summands}.
Each summand symbolically represents a partial transition relation between the current and the next state for a multi-set of action labels.
Let $\names$ be a set of process \emph{names}.

\begin{definition}
  A \emph{linear process equation} (LPE) is an equation of the form:
  \[P(d : D) = \sum_{e_0 : E_0} c_0 \rightarrow \multiact_0 \sequential P(g_0) 
  + ~\ldots~
  + \sum_{e_n : E_n} c_n \rightarrow \multiact_n \sequential P(g_n)
  \]
  Where $P \in \names$ is the process \emph{name} and $d$ is the process parameter.
  Furthermore, for $0 \leq i \leq n$ it holds that:
  \begin{itemize}  
    \item $E_i$ is a sort over which \emph{sum} variable $e_i$ (where $e_i \neq d$) ranges,

    \item $c_i$ is a boolean expression such that $\freevars(c_i) \subseteq \set{d, e_i}$ defining the \emph{enabling condition}, and

    \item $\multiact_i$ is a multi-action $\tau$ or $a^1_i(f^1_i) | \ldots | a^{n_i}_i(f^{n_i}_i)$ such that $a^k_i \in \events$ and $f^k_i$ is an expression of sort $D_{a^k_i}$ such that $\freevars(f^k_i) \subseteq \set{d, e_i}$, for $1 \leq k \leq n_i$, and

    \item $g_i$ is an \emph{update} expression such that $\freevars(g_i) \subseteq \set{d, e_i}$ of sort $D$.
  \end{itemize}
\end{definition}

The $+$ operator represents a non-deterministic choice among its operands, and the sum operator, denoted by $\sum$, expresses a non-deterministic choice for values of the sum variable.
Note that the sum operator acts as a binder for the variable $e_i$.
The sum operator is omitted whenever variable $e_i$ does not occur freely within the condition, action and update expressions.
We use $\bigplus_{i \in I}$ for a finite set of \emph{indices} $I \subseteq \naturalnumbers$ as a shorthand for a number of summands.

We often consider LPEs where the parameter sort $D$ represents a \emph{vector} of a length $n + 1$; in that case we write $d_0 : D_0, \ldots, d_n : D_n$ to indicate that there are $n + 1$ parameters such that each $d_i$ has sort $D_i$ for $0 \leq i \leq n$.
Similarly, we also generalise the action sorts and the sum operator in LPEs, where we permit ourselves to write $a(e_0, \ldots, e_k)$ and $\sum_{e_0:E_0, \ldots, e_l:E_l}$, respectively.

The \emph{operational semantics} of an LPE are defined by a mapping to an LTS.
Given a substitution $\subst = [x_0 \gets \initval_0, \ldots, x_n \gets \initval_n]$ and an expression $e$ we use $\subst(e)$ to denote the expression $e$ where each occurrence of variable $x_i$ is syntactically replaced by the \emph{closed} expression $\initval_i$, for $0 \leq i \leq n$.
We assume the usual principle of substitutivity where for all variables $x$, expressions $f$ and closed expressions $g$ and $h$ it holds that if $g \approx h$ then $[x \gets g](f) \approx [x \gets h](f)$.
Let $\linearprocess$ be the set of symbols $P(\initval)$ such that $P(d : D) = \phi_P$, for any $P \in \names$, is an LPE and $\initval$ is a closed expression of sort $D$.

\begin{definition}\label{def:lpesemantics}
  Let $P(d : D) = \bigplus_{i \in I} \sum_{e_i : E_i} c_i \rightarrow \multiact_i \sequential P(g_i)$ be an LPE and let $\initval : D$ be a closed expression. The semantics of $P(\initval)$, denoted by $\interpret{P(\initval)}$, is an LTS $(\linearprocess, P(\initval), \semactions, \transitions)$ for which the following holds.
  For all values $j \in I$, closed expressions $l : E_j$ and closed expressions $\initval' : D$ such that $\subst = [d \gets \initval', e_j \gets l]$ there is a transition $P(\initval') \xrightarrow{\interpret{\subst(\multiact_j)}} P(\subst(g_j))$ iff $\interpret{\subst(c_j)} = \textbf{\true}$.
\end{definition}

We refer to the reachable part of the LTS that is the interpretation of an LPE with a given closed expression as the \emph{state space} of that LPE.
In the interpretation of an LPE a syntactic substitution is applied on the update expressions to define the reached state.
This means that different closed syntactic expressions which correspond to the same semantic object, \eg, $1 + 1$ and $2$ for our assumed sort $\natsort$, result in different states.
However, such states are always strongly bisimilar as formalised below.

\TW{The proof of the lemma may be omitted for the conference version; also the text following the proof can be condensed/removed}
\begin{lemma}\label{lemma:bisimulation}
	Given an LPE $P(d : D) = \phi_P$ and a closed expression $\initval : D$.
	For all closed expressions $e, e' : D$ such that $\interpret{e \approx e'} = \textbf{\true}$ we have $\interpret{P(e)} \bisim \interpret{P(e')}$.	
\end{lemma}
\begin{proof}
	We can show that the smallest relation $\related$ such that $P(e) \related P(e')$ for all closed expressions $e, e' : D$ such that $e \approx e'$ is a strong bisimulation relation.
	This follows essentially from the principle of substitutivity and the definition of an LPE.
\end{proof}

For any given state space we can therefore consider a \emph{representative} state space where for each state a unique closed expression is chosen that is data equivalent.
In an implementation it is very natural to generate this representative LTS directly, for example by always reducing expressions to the normal form if the data specification is based on a (terminating and confluent) term rewrite system.
In examples we always consider the representative state space.

\begin{example}\label{example:delayedmachine}
	Consider the LPE which models a machine that can be toggled with a delay of three counts.
	Whenever the counter reaches zero, \ie, $n$ is zero, it can be toggled again after which it counts down three times.
	\begin{align*}
		\procmachine(n : \natsort, s : \boolsort) 
	    & = (n > 0) \rightarrow \actcount \sequential \procmachine(n - 1, s) \\
	    & + (n \approx 0) \rightarrow \acttoggle(s) \sequential \procmachine(3, \neg s)
	\end{align*}
 	Note that the sum operator has been omitted, because only parameter $n$ occurs as a free variable in the expressions.
	A representative state space of the machine that is initially off, defined by $\interpret{\procmachine(0, \false)}$, is shown in Figure~\ref{figure:delayedmachine}.
  Initially, $n$ is zero and therefore $\interpret{0 \approx 0}$ is $\textbf{\true}$ and there is a transition labelled $\interpret{\acttoggle(\true)} = \multiset{\acttoggle(\textbf{\true})}$ to the state $\procmachine(3, \true)$.
  The other summand does not result in transitions for state $\procmachine(0, \false)$, because $\interpret{0 > 0}$ is $\textbf{\false}$.
  Similarly, there is an outgoing transition labelled $\interpret{\actcount} = \multiset{\actcount}$ for $\procmachine(3, \true)$, because $\interpret{3 > 0}$ is $\textbf{\true}$. Again, there is no other outgoing transition for state $\procmachine(3, \true)$ as $\interpret{3 \approx 0}$ is $\textbf{\false}$.
  The other transitions are derived in a similar way. 
\end{example}

\ML{Section 2.3 is hard to read. It lacks textual, intuitive explanations.}
\subsection{A Process Algebra of Communicating Linear Process Equations}
\label{sec:clpe}

We define a simple process algebra to express parallelism and interaction of LPEs.
Let $\comms$ be the set of \emph{communication} expressions $a_0 | \ldots | a_n \rightarrow c$ where $a_i, c \in \events$ for $0 \leq i \leq n$ are action labels.

\begin{definition}
  The process algebra is defined as follows:
  \begin{equation*}
    S ::= \communication_{C}(S) ~\mid~ \allow_A(S) ~\mid~ \hide_H(S) ~\mid~ S \parallel S ~\mid~ P(\initval)
  \end{equation*}
  Where $A \subseteq 2^{\events \rightarrow \naturalnumbers}$ is a non-empty finite set of finite multi-sets of action labels, $H \subseteq \events$ is a non-empty finite set of action labels and $C \subseteq \comms$ is a finite set of \emph{communications}.
  Finally, we have $P(\initval) \in \linearprocess$.
\end{definition}

The set $\simpleprocess$ contains all expressions of the process algebra.
The operators describe \emph{communication} ($\communication_{C}$), \emph{action allowing} ($\allow_A$), \emph{action hiding} ($\hide_H$) and \emph{parallel composition} ($\parallel$).
Finally, the elementary objects are the processes, defined as LPEs.

First, we introduce several auxiliary functions on $\semactions$ that are used to define the semantics of expressions in $\simpleprocess$.

\begin{definition}
  Given $\action \in \semactions$ we define $\actcomm{C}$, where $C \subseteq \comms$, as follows:
  \begin{align*}
    \actcomm{\emptyset}(\action) &= \action \\
    \actcomm{C}(\action) &= \actcomm{C \setminus C_1}(\actcomm{C_1}(\action))  \text{ for } C_1 \subset C \\
    \actcomm{\set{a_0 | \ldots | a_n \rightarrow c}}(\action) &=
      \begin{cases}
        \parbox{0.6\columnwidth}{$\multiset{c(\sem{d})} + \actcomm{\{a_0 | \ldots | a_n \rightarrow c\}}(\action - \multiset{a_0(\sem{d}), \ldots, a_n(\sem{d})})$} \\
        		    \quad\,\,\,\, \text{if } \multiset{a_0(\sem{d}), \ldots, a_n(\sem{d})} \subseteq \action \\
        \action \quad \text{otherwise}
      \end{cases}
  \end{align*}
\end{definition}

For this function to be well-defined we require that the labels in the left-hand sides of the communications should not \emph{overlap}.
Furthermore, the action label on the right-hand side must not occur in any \emph{other} left-hand side.
\MLadd{For example $\actcomm{\set{a|b\rightarrow c}}(a|d|b) = c|d$ according to the definition, but $\actcomm{\set{a|b\rightarrow c, a|d\rightarrow c}}(a|d|b)$ and $\actcomm{\set{a|b\rightarrow c,c\rightarrow d}}(a|d|b)$ are not allowed.}

\begin{definition}
  Given $\omega \in \semactions$ we define $\acthide{H}(\action)$, where $H \subseteq \events$ is a set of action labels, such that $\acthide{H}(\omega) = \omega'$ where:
  \begin{equation*}
  	\action'(a(\sem{d})) = 
  	\begin{cases}
  		0 &\text{if } a \in H \\
  		\action(a(\sem{d})) &\text{otherwise}
  	\end{cases}
  \end{equation*}
\end{definition}

Given a multi-action $\multiact$ we define $\nodata{\multiact}$ to obtain the multi-set of action labels, \eg, $\nodata{a(3)|b(5)} = \multiset{a, b}$.
Formally, $\nodata{a(e)} = \multiset{a}$, $\nodata{\tau} = \multiset{}$ and $\nodata{\alpha|\beta} = \nodata{\alpha} + \nodata{\beta}$.
We define $\nodata{\action}$ for $\action \in \semactions$ in a similar way.

Using these functions we can define the semantics, using \emph{structured operational semantics}, as an LTS.
  
\begin{definition}
  The operational semantics of an expression $Q$ of $\simpleprocess$, denoted $\interpret{Q}$, are defined by the corresponding LTS $(\simpleprocess, Q, \semactions, \transitions)$ with its transition relation defined by the rules below and the transition relation given in Definition~\ref{def:lpesemantics} for each expression in $\linearprocess$.
  For any $\action \in \semactions$ and $P, P', Q, Q'$ expressions of $\simpleprocess$:
  
  \begin{center}
  \begin{tabular}{c c}
    $\inferrule*[Left=Com]{P \xrightarrow{\action} P' \\ C \subseteq \comms}{\communication_{C}(P) \xrightarrow{\actcomm{C}(\action)} \communication_{C}(P')}$ &
    \quad\quad\quad $\inferrule*[Left=Allow]{P \xrightarrow{\action} P' \\ A \subseteq 2^{\events \rightarrow \naturalnumbers} \\ \nodata{\action} \in A}{\allow_A(P) \xrightarrow{\action} \allow_A(P')}$ \\[4ex]

    $\inferrule*[Left=Hide]{P \xrightarrow{\action} P' \\ H \subseteq \events}{\hide_H(P) \xrightarrow{\acthide{H}(\action)} \hide_H(P')}$
    & $\inferrule*[Left=Par]{P \xrightarrow{\action} P' \quad Q \xrightarrow{\action'}  Q'}{P \parallel Q \xrightarrow{\action \,+\, \action'} P' \parallel Q'}$ \\[4ex]

    $\inferrule*[Left=ParR]{Q \xrightarrow{\action} Q'}{P \parallel Q \xrightarrow{\action} P \parallel Q'}$
    & $\inferrule*[Left=ParL]{P \xrightarrow{\action} P'}{P \parallel Q \xrightarrow{\action} P' \parallel Q}$  
  \end{tabular}
  \end{center}
\end{definition}

\section{Decomposition}\label{section:decomposition}

We are interested in decomposing an LPE into $n$ LPEs, where the latter are referred to as \emph{components}, such that each component contains a subset of the original parameters.
This decomposition is considered \emph{valid} iff the original state space is strongly bisimilar to the state space of these components under a suitable context.

\begin{definition}\label{def:valid_decomposition}
  Let $P(\vec{d} : \vec{D}) = \phi$ be an LPE and $\vec{\initval} : \vec{D}$ a closed expression.
  The LPEs $P_0(\subvector{\vec{d}}{I_0} : \subvector{\vec{D}}{I_0}) = \phi_0$ to $P_n(\subvector{\vec{d}}{I_n} : \subvector{\vec{D}}{I_n}) = \phi_n$, for indices $I_0, \ldots, I_n \subseteq \naturalnumbers$, are a valid \emph{decomposition} of $P$ and $\vec{\initval}$ under a context $C$ iff:  \begin{equation*}
    \interpret{P(\vec{\initval})} \bisim \interpret{\mathsf{C}[P_0(\subvector{\vec{\initval}}{I_0}) \parallel \ldots \parallel P_n(\subvector{\vec{\initval}}{I_n})]}
  \end{equation*}
  Where $\mathsf{C}[P_0(\subvector{\vec{\initval}}{I_0}) \parallel \ldots \parallel P_n(\subvector{\vec{\initval}}{I_n})]$ is an expression in $\simpleprocess$.
  We refer to the expression $\mathsf{C}[P_0(\subvector{\vec{\initval}}{I_0}) \parallel \ldots \parallel P_n(\subvector{\vec{\initval}}{I_n})]$ as the \emph{composition}.
\end{definition}

\MLadd{In our case the context $\mathsf{C}$ is an expression with operators from $\simpleprocess$ that is used to define the synchronisation between the individual components, see for example the composition expression in Definition~\ref{def:refinedcleavecomposition}.}
The primary benefit of a valid decomposition is that a state space that is equivalent to the original state space can be obtained as follows.
First, the state space of each component is derived separately.
Then the result of the composition expression can be derived from the component state spaces based on the rules of the operational semantics.
Furthermore, the component state spaces can be minimised modulo an equivalence relation that is a congruence with respect to the operators of $\simpleprocess$ before deriving the results of the composition expression.
Strong bisimilarity is known to be a congruence with respect to our operators.
This is referred to as \emph{compositional minimisation}.

For the remainder of this paper we consider a decomposition technique that results in exactly two components.
First, we present an example to show that a valid decomposition can be achieved using the presented process algebra and which shows that not every decomposition is necessarily useful for state space exploration.
This decomposition exploits the fact that $\tau$ represents the empty multi-action, which means that $\multiact | \tau = \multiact$ holds for all multi-actions $\multiact$.

\begin{example}\label{example:naivecleave}
	Consider the LPE of Example~\ref{example:delayedmachine} again.
  For the decomposition we introduce the two components shown below.
	\begin{align*}
		\procmachine_V(n : \natsort)
		    & = \sum_{s : \boolsort} (n > 0) \rightarrow  \actcount | \actsync^0_V(n, s) \sequential \procmachine_V(n - 1) \\
		    & + \sum_{s : \boolsort} (n \approx 0) \rightarrow \tau | \actsync^1_V(n, s) \sequential \procmachine_V(3) \\
 		\procmachine_W(s : \boolsort)
		    & = \sum_{n : \natsort} (n > 0) \rightarrow \tau | \actsync^0_W(n, s) \sequential \procmachine_W(s) \\
		    & + \sum_{n : \natsort} (n \approx 0) \rightarrow \acttoggle(s) | \actsync^1_W(n, s) \sequential \procmachine_W(\neg s)
	\end{align*}
  Each component describes part of the behaviour where the parameter value, either $n$ or $s$, is known.  
  However, the value of the other parameter is unknown.
  To cater for this, we add it as a sum variable.
  Consider the state space of $\procmachine_V(0)$ shown below.
  It describes the behaviour of $M$ where the value of $n$ is known, but at each state it allows both values of $s$.
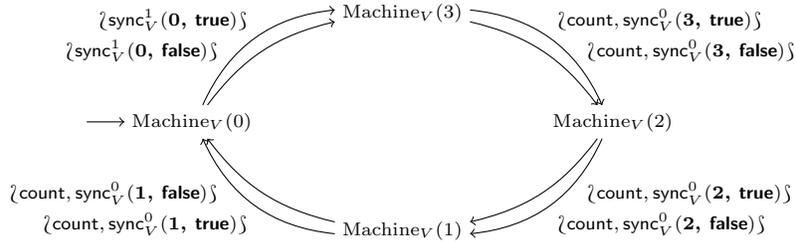
  \begin{center}
 	\begin{tikzpicture}[->]
 		\scriptsize
	  \node (s0) {$\procmachine_V(0)$};
	  \node (s0init) [left= 0.5cm of s0] {};
	  \node (s1) [above right= of s0] {$\procmachine_V(3)$};
	  \node (s2) [below right= of s1] {$\procmachine_V(2)$};
	  \node (s3) [below right= of s0] {$\procmachine_V(1)$};
 	
 		\draw (s0init) edge (s0);
	  \draw (s0) edge[bend left] node[above left] {$\multiset{\actsync^1_V(\textbf{0, \true})}$} (s1);
	  \draw (s0) edge[bend left=20] node[left=.5cm] {$\multiset{\actsync^1_V(\textbf{0, \false})}$} (s1);
	  
	  \draw (s1) edge[bend left] node[above right] {$\multiset{\actcount,\actsync^0_V(\textbf{3, \true})}$} (s2);	  
	  \draw (s1) edge[bend left=20] node[right=0.5cm] {$\multiset{\actcount,\actsync^0_V(\textbf{3, \false})}$} (s2);
	  
 	  \draw (s2) edge[bend left] node[below right] {$\multiset{\actcount,\actsync^0_V(\textbf{2, \false})}$} (s3);
	  \draw (s2) edge[bend left=20] node[right=.5cm] {$\multiset{\actcount,\actsync^0_V(\textbf{2, \true})}$} (s3);
	  
 	  \draw (s3) edge[bend left] node[below left] {$\multiset{\actcount,\actsync^0_V(\textbf{1, \true})}$} (s0);
 	  \draw (s3) edge[bend left=20] node[left=.5cm] {$\multiset{\actcount,\actsync^0_V(\textbf{1, \false})}$} (s0);
 	\end{tikzpicture}
  \end{center}  
  The synchronisation actions make the values of the parameters, which are chosen non-deterministically for the unknown parameters, visible in the behaviour.
  This can be used to achieve a valid decomposition by enforcing the synchronisation of these actions as follows:
  \begin{align*}
    & \allow_{\{\multiset{\acttoggle}, \multiset{\actcount}\}}( \hide_{\{\actsync^0, \actsync^1\}}( \\ & \quad \communication_{\{\actsync^0_V|\actsync^0_W \rightarrow \actsync^0, \actsync^1_V|\actsync^1_W \rightarrow \actsync^1\}}(\procmachine_V(0) \parallel \procmachine_W(\false))))
  \end{align*}
  
  Unfortunately the state space of $\procmachine_W(\false)$ is infinitely branching from its initial state and it has no finite state space that is strongly bisimilar to it.
  Therefore, the previously described state space construction cannot be applied to this decomposition.
  However, it can be verified that this is a valid decomposition.
\end{example}

For the purpose of state space exploration the effectiveness of a valid decomposition is determined by the size, \ie, the sum of the number of states and transitions, of the state space corresponding to each component.
As a minimum requirement the size of each component state space should be smaller than the size of the original state space.
Furthermore, it would be considered useful whenever the composition of the minimised components is smaller than the original state space.

\ML{Sections 3.1 and 3.2 should be reworked carefully to be made clearer.}
\subsection{Separation Tuples}

Example~\ref{example:naivecleave} hinted at the construction that we use to achieve a valid decomposition.
However, as we have seen this decomposition could not yet be used for the compositional state space construction.
Furthermore, we need to consider the restrictions that the resulting components should satisfy in order to be a valid decomposition in the general case.

To obtain a useful decomposition it can be beneficial to reduce the number of parameters that occur in the synchronisation actions, because these actions become visible as transitions in the state spaces of the individual components.
Furthermore, in some cases we can actually remove the synchronisation for summands completely.
For instance, in the first summand of $\procmachine$ in Example~\ref{example:delayedmachine} we can observe that the value of parameter $s$ remains unchanged and the condition is only an expression containing parameter $n$.
Therefore, we could allow component $\procmachine_V$ to result in a transition labelled with $\actcount$ without a corresponding summand in $\procmachine_W$ that synchronises the values of $s$ and $n$ unnecessarily.
We refer to these kind of summands as \emph{independent} summands.

An independent summand can result in transitions without synchronising with the other component in the composition expression.
However, this might introduce an issue when the action labels in all action expressions are \emph{not} disjoint.
For example, consider an LPE that contains two summands with $a$ as action expression that are independent in different components and another summand with $a|a$ as action expression.
Then $a|a$ would be allowed in the composition expression to ensure that these transitions can occur, but then the independent summands could also result in a simultaneous transition due to rule $\textsf{Par}$, which was not a possibility in the original LPE.
To prevent this issue we introduce a \emph{tag} action label and only allow independent actions with a single tag in the composition.

First, we present several restrictions on the structure of the component LPEs.
In the following definition there is a set of indices $\Iind$ to partition the summands between dependent and \emph{independent} summands.
Furthermore, there is a set of indices $\Isub$ to indicate summands that are present in this component.
Finally, we allow the condition, action and synchronisation expressions to be chosen freely, which can be used to further reduce the amount of parameter synchronisation.
We use indexed sets for these expressions such that the index corresponds to the index of the summand, where for each element we indicate the index by a subscript.

\ML{No intuition what $U$, $K$ and $J$ represent}
\begin{definition}\label{def:derived}
  Let $P(\vec{d} : \vec{D}) = \bigplus_{i \in I} \sum_{e_i : E_i} c_i \rightarrow \multiact_i \sequential P(\vec{g}_i)$ be an LPE.
  Let $(P, U, \Iind, \Isub, c^U, \multiact^U, h^U)$ be a \emph{separation} tuple such that $U \subseteq \naturalnumbers$ is a set of parameter indices and $\Iind \subseteq \Isub \subseteq I$ are two sets of summand indices.
  Furthermore, $c^U, \multiact^U$ and $h^U$ are indexed sets of condition, action and update expressions respectively such that for all $i \in (\Isub \setminus \Iind)$ it holds that $\freevars(c^U_i) \cup \freevars(\multiact^U_i) \cup \freevars(\vec{h^U_i}) \subseteq \flatten{\vec{d}} \cup \{e_i\}$.
  Finally, for all $i \in K$ it holds that $\freevars(c_i) \cup \freevars(\multiact_i) \cup \freevars(\subvector{\vec{g_i}}{U}) \subseteq \flatten{\subvector{\vec{d}}{U}} \cup \{e_i\}$.
   
  The separation tuple induces an LPE, where $U^c = \naturalnumbers \setminus U$, as follows:
  \begin{align*}
    P_{U}(\subvector{\vec{d}}{U} : \subvector{\vec{D}}{U}) = &\bigplus_{i \in (\Isub \setminus \Iind)} \sum_{e_i : E_i, \subvector{\vec{d}}{U^c} : \subvector{\vec{D}}{U^c}} \\ 
    & \,\,\,\quad\quad\quad c^U_i \rightarrow \multiact^U_i|\textsf{sync}^i_U(\vec{h^U_i}) \sequential P_U(\subvector{\vec{g_i}}{U}) \\
  + &~~ \bigplus_{i \in K} \sum_{e_i : E_i} c_i \rightarrow \multiact_i | \textsf{tag} \sequential P_U(\subvector{\vec{g_i}}{U})
  \end{align*}
  We assume that action labels $\actsync^i_V$ and $\actsync^i_W$, for any $i \in I$, and label $\acttag$ does not occur in $\multiact_j$, for any $j \in I$, to ensure that these action labels are fresh.
\end{definition}

The composition expression for these components is a generalisation of the composition expression presented in Example~\ref{example:naivecleave} where $\acttag$ action labels are hidden at the highest level.

\begin{definition}\label{def:refinedcleavecomposition}
  Let $P(\vec{d} : \vec{D}) = \bigplus_{i \in I} \sum_{e_i : E_i} c_i \rightarrow \multiact_i \sequential P(\vec{g}_i)$ be an LPE and $(P, V, \Iind_V, \Isub_V, c^V, \multiact^V, h^V)$ and $(P, W, \Iind_W, \Isub_W, c^W, \multiact^W, h^W)$ be separation tuples.
  Let $P_V(\subvector{\vec{d}}{V} : \subvector{\vec{D}}{V}) = \phi_V$ and $P_W(\subvector{\vec{d}}{W} : \subvector{\vec{D}}{W}) = \phi_W$ be the induced LPEs according to Definition~\ref{def:derived}.
  Let $\initval : \vec{D}$ be a closed expression.
  Then the composition expression is defined as:
  \begin{align*}
    &\hide_{\{\acttag\}} (
    	\allow_{\{\nodata{\multiact_i} \mid i \in I\} \cup \{\nodata{\multiact_i|\acttag} \mid i \in \Iind_V\} \mid i \in \Iind_W\}\}}( \\
    		&\hide_{\{\actsync^i \mid i \in I\}} (\communication_{\{ \actsync^i_V \mid \actsync^i_W \rightarrow \actsync^i \mid i \in I\}} (P_V(\subvector{\vec{\initval}}{V}) \parallel P_W(\subvector{\vec{\initval}}{W})))
    	  )
      )
  \end{align*}
\end{definition}

We revisit Example~\ref{example:delayedmachine} to illustrate a valid decomposition that fits the given structure and unlike Example~\ref{example:naivecleave} allows for reduced synchronisation.

\begin{example}\label{example:refinedcleave}
	Consider the LPE presented in Example~\ref{example:delayedmachine} again.	
	We obtain components $P_V$ and $P_W$ for the separation tuples $(P, V, \{0\}, \{0,1\}, \{(n > 0)_0, (n \approx 0)_1\}, \{(\actcount|\acttag)_0, (\tau)_1\}, \{\vector{}_1\})$ and $(P, W, \emptyset, \{1\}, \{\true_1\}, \{\acttoggle(s)_1\}, \{\vector{}_1\})$ respectively.
  \begin{align*}
    P_V(n : \natsort) &= (n > 0) \rightarrow \actcount|\acttag \sequential P_V(n - 1) \\
      &+ (n \approx 0) \rightarrow \actsync^1_V \sequential P_V(3) \\
    P_W(s : \boolsort) &= \true \rightarrow \acttoggle(s) | \actsync^1_W \sequential P_W(\neg s)
  \end{align*}
  The state spaces of components $P_V(0)$ and $P_W(\false)$ are shown below.
  
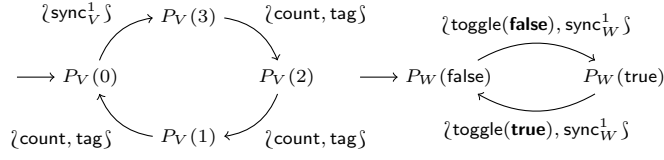
\begin{center}
	\begin{tikzpicture}[->]
		\scriptsize
	  \node (s0) {$P_V(0)$};
	  \node (s0init) [left= 0.5cm of s0] {};
	  \node (s1) [above right= 0.5cm of s0] {$P_V(3)$};
	  \node (s2) [below right= 0.5cm of s1] {$P_V(2)$};
	  \node (s3) [below right= 0.5cm of s0] {$P_V(1)$};
	
		\draw (s0init) edge (s0);
	  \draw (s0) edge[bend left] node[above left] {$\multiset{\actsync^1_V}$} (s1);	  
	  \draw (s1) edge[bend left] node[above right] {$\multiset{\actcount,\acttag}$} (s2);
 	  \draw (s2) edge[bend left] node[below right] {$\multiset{\actcount,\acttag}$} (s3);  
 	  \draw (s3) edge[bend left] node[below left] {$\multiset{\actcount,\acttag}$} (s0);
 	  
	  \node (t0) [right=1cm of s2] {$P_W(\false)$};
	  \node (t0init) [left= 0.5cm of t0] {};
	  \node (t1) [right= of t0] {$P_W(\true)$};
  	
  	\draw (t0init) edge (t0);
	  \draw (t0) edge[bend left] node[above] {$\multiset{\acttoggle(\textbf{\false}), \actsync^1_W}$} (t1);	  
	  \draw (t1) edge[bend left] node[below] {$\multiset{\acttoggle(\textbf{\true}), \actsync^1_W}$} (t0);
	\end{tikzpicture}
  \end{center}
  We obtain the following composition according to Definition~\ref{def:refinedcleavecomposition}:
  \begin{align*}
  	\hide_{\{\acttag\}}(\allow_{\{\multiset{\acttoggle},\multiset{\actcount},\multiset{\actcount,\acttag}\}} (\hide_{\{\actsync^0,\actsync^1\}} (\\
  	\communication_{
  		\{\actsync^0_V|\actsync^0_W \rightarrow \actsync^0, \actsync^1_V|\actsync^1_W \rightarrow \actsync^1\}
  		} (P_V(0) \parallel P_W(\false)))))
  \end{align*}
  One can verify that the state space of this expression is strongly bisimilar to the state space of $\procmachine(0, \false)$ shown in Example~\ref{example:delayedmachine}.
  As shown above the state space of $P_V(0)$ has four states and transitions, and the state space of $P_W(\false)$ has two states and transitions, which are both smaller than the original state space.
  Their composition is exactly the same size as the original state space, where no further minimisation can be achieved.
\end{example}

\ML{One reviewer did not understand that decomposition can achieve more then just identifying independent summands.}
\subsection{Cleave Correctness Criteria}

Not every decomposition which satisfies Definition~\ref{def:refinedcleavecomposition} yields a valid decomposition.
For example, replacing the condition expression in Example~\ref{example:refinedcleave} of the summand in $P_W$ by $\false$ would not result in a valid decomposition.
Therefore, we need to consider restrictions that should be imposed on the components such that the result of the composition expression defined in Definition~\ref{def:refinedcleavecomposition} is \emph{always} a valid decomposition, which means that the composition should be strongly bisimilar to the original state space for the given initial values.
We essentially employ restrictions to preserve a relation between each original state and the two states of the components where the parameters have the same value.

Let us consider any decomposition according to Definition~\ref{def:refinedcleavecomposition}.
We provide an intuition for each of the requirements that is stated in Definition~\ref{def:requirements}.
First of all, we need to ensure that every summand of the LPE either occurs as an independent summand in one of the components or it occurs in both components.
This is stated by requirement \REQ{1}.
Furthermore, this requirement ensures that the condition, action and update expressions for summands with indices in $(\Isub_V \cap \Isub_W)$ are defined in requirements \REQ{3} and \REQ{4}.

In Definition~\ref{def:refinedcleavecomposition} we see that summands with indices in $K$ have condition, action and update expressions that only have the parameters in $\subvector{\vec{d}}{V}$ and the summand variable as free variables.
Furthermore, to consider a summand as \emph{independent} we require that the update expressions in $\subvector{\vec{d}}{W}$ do not modify the parameters.
This is essential as otherwise these updates would be omitted, which would almost certainly not yield a valid decomposition.
Therefore, we introduce requirement \REQ{2} to ensure that elements of $K_V$ and $K_W$ are independent summands.

For the summands with an index in $(\Isub_V \cap \Isub_W)$ we need to consider when the synchronisation should occur and when it is not allowed.
Whether it is allowed depends on the outgoing transitions that occur due to the corresponding original summand, \ie, the summand in the original LPE with the same index.
If there is an outgoing transition due to such a summand then both components must result in a transition where the synchronisation action can communicate.
This means that their conditions must be true, the synchronisation vectors ($\vec{h}^V$ and $\vec{h}^W$) must have equal values and the resulting action label must be equal to the original action label
These are exactly the requirements stated by \REQ{3}.

Similarly, whenever both components have outgoing transitions that could synchronise then there must be a corresponding outgoing transition in the original state space.
Here, the complication is that the values for the sum variables, which are both the original sum variable and the values of the other parameters, can be chosen non-deterministically, and therefore requirement \REQ{4} ensures that for any choice of these variables there is a corresponding outgoing transition in the original state.  

\begin{definition}\label{def:requirements}
  Let $P(\vec{d} : \vec{D}) = \bigplus_{i \in I} \sum_{e_i : E_i} c_i \rightarrow \multiact_i \sequential P(\vec{g}_i)$ be an LPE and $(P, V, \Iind_V, \Isub_V, c^V, \multiact^V, h^V)$ and $(P, W, \Iind_W, \Isub_W, c^W, \multiact^W, h^W)$ be separation tuples as defined in Definition~\ref{def:derived}.
	The two separation tuples are a \emph{cleave} of $P$ iff the following requirements hold.
	
	\begin{enumerate}		
		\item[\REQ{1}] $\Isub_V = I \setminus \Iind_W$ and $\Isub_W = I \setminus \Iind_V$.
		
		\item[\REQ{2}] For all $r \in \Iind_V$ it holds that $\subvector{\vec{g_r}}{W} = \subvector{\vec{d}}{W}$, and for all $r \in \Iind_W$ it holds that $\subvector{\vec{g_r}}{V} = \subvector{\vec{d}}{V}$.
		
		\item[\REQ{3}] For all $r \in (\Isub_V \cap \Isub_W)$, closed expressions $\vec{\initval} : \vec{D}$ and $l : E_r$ such that $\subst = [\vec{d} \gets \vec{\initval}, e_r \gets l]$ it holds that if $\interpret{\subst(c_r)}$ then:		
		  \begin{itemize}
		    \item $\interpret{\subst(c^V_r \land c^W_r)}$, and
		
		    \item $\interpret{\subst(\vec{h^V_r} \approx \vec{h^W_r})}$, and
		
		    \item $\interpret{\subst(\multiact^V_r | \multiact^W_r \approx \multiact_r)}$.
		  \end{itemize}
		  
		\item[\REQ{4}] For all $r \in (\Isub_V \cap \Isub_W)$, closed expressions $\vec{\initval}, \vec{\initval'}, \vec{\initval''} : \vec{D}$ and $l, l' : E_r$ such that $\subst = [\subvector{\vec{d}}{V} \gets \subvector{\vec{\initval}}{V}, \subvector{\vec{d}}{W \setminus V} \gets \subvector{\vec{\initval'}}{W \setminus V}, e_r \gets l]$ and $\sigma' = [\subvector{\vec{d}}{W} \gets \subvector{\vec{\initval}}{W}, \subvector{\vec{d}}{V \setminus W} \gets \subvector{\vec{\initval''}}{V \setminus W}, e_r \gets l']$ the following holds.
		  If both $\interpret{\subst(c^V_r) \land \sigma'(c^W_r)}$ and $\interpret{\subst(\vec{h^V_r}) \approx \sigma'(\vec{h^W_r})}$ hold then there is a closed expression $l'' : E_r$ such that for the substitution $\rho = [\vec{d} \gets \vec{\initval}, e_r \gets l'']$ it holds that:      
      
      \begin{itemize}
        \item $\interpret{\rho(c_r)}$, and
	
        \item $\interpret{\subst(\multiact^V_r)|\subst'(\multiact^W_r)} = \interpret{\rho(\multiact_r)}$, and
	
        \item $\interpret{\subst(\subvector{\vec{g_r}}{V})} = \interpret{\rho(\subvector{\vec{g_r}}{V})}$, and
        
        \item $\interpret{\subst'(\subvector{\vec{g_r}}{W})} = \interpret{\rho(\subvector{\vec{g_r}}{W})}$.
      \end{itemize}
	\end{enumerate}
\end{definition}

Note that requirements \REQ{3} and \REQ{4} of Definition~\ref{def:requirements} are stated on the semantics of the condition, action and update expressions.
In practice, we would need to effectively approximate these correctness requirements using static analysis.
For example the requirements on the condition expression can be approximated by using the syntactic conjunctions that are present in it.
However, how precise and efficient this static analysis can be is left as future work. 
The correctness of the cleave is established by the following theorem.

\begin{restatable}{theorem}{cleavecorrectness}\label{theorem:refinedcleave}
	The composition expression defined in Definition~\ref{def:refinedcleavecomposition} where the two separation tuples are a cleave according to Definition~\ref{def:requirements} is a valid decomposition as defined in Definition~\ref{def:valid_decomposition}.
\end{restatable}
\begin{proof}
  Let $\related$ be the smallest relation such that for any closed expressions $\vec{\initval'} : \vec{D}$ the following holds.
  \begin{align*}
  P(\vec{\initval}') ~\related~ &\hide_{ \{\acttag\}} (\allow_{\{\nodata{\multiact_i} \mid i \in I\} \cup \{\nodata{\multiact_i|\acttag} \mid i \in \Iind_V\} \mid i \in \Iind_W\}\}}( \hide_{\{\actsync^i \mid i \in I\}} 
    \\ & (\communication_{\{ \actsync^i_V \mid \actsync^i_W \rightarrow \actsync^i \mid i \in I\}} (P_V(\subvector{\vec{\initval}}{V}) \parallel P_W(\subvector{\vec{\initval}}{W})))))
  \end{align*}
  We prove that $\related$ is a strong bisimulation \emph{up to bisimilarity} $\bisim$.
  The essential observation is that due to parameter synchronisation the original state vectors can always be traced to the two states (of $P_V$ and $P_W$) that carry the same values, and vice versa.
  The requirements ensure that (the combination of) expressions evaluate to the same values as the corresponding original values.
  The full proof can be found in Appendix~\ref{appendix:refinedcleave}.
\end{proof}

Observe that the decompositions obtained in Example~\ref{example:naivecleave} and Example~\ref{example:refinedcleave} are cleaves.
However, we can also observe that the decomposition in Example~\ref{example:naivecleave} is not a particularly useful cleave for the purpose of compositional minimisation.
In this case, as shown by Example~\ref{example:refinedcleave} we could avoid the infinite branching of $\procmachine_W(\false)$ by reducing the amount of synchronisation, but this might not always be possible.
Therefore, we present an alternative technique to restrict the behaviour of the resulting components.

\section{State Invariants}\label{section:state_invariant}

One way to restrict the behaviour of the components is to strengthen the condition expressions of each summand to avoid certain outgoing transitions.
We show that so-called \emph{state invariants}~\cite{BezemG94} can be used for this purpose.
\MLadd{These state invariants are typically formulated by the user based on intuition of the model behaviour.}

\begin{definition}
	Given an LPE $P(d : D) = \bigplus_{i \in I} \sum_{e_i : E_i} c_i \rightarrow \multiact_i \sequential P(g_i)$.
	A boolean expression $\psi$ such that $\freevars(\psi) \subseteq \{d\}$ is called a \emph{state invariant} iff the following holds:
	for all $i \in I$ and closed expressions $\initval : D$ and $l : E_i$ such that $\interpret{[d \gets \initval, e_i \gets l](c_i \land \psi)}$ holds then $\interpret{[d \gets [d \gets \initval, e_i \gets l](g_i)](\psi)}$ holds as well.
\end{definition}

The essential property of a state invariant is that whenever it holds for the initial state it is guaranteed to hold for all reachable states in the state space.
This follows relatively straightforward from its definition.
Next, we define a \emph{restricted} LPE where (some of) the condition expressions are strengthened with a boolean expression.

\begin{definition}\label{def:restricted_lpe}
	Given an LPE $P(d : D) = \bigplus_{i \in I} \sum_{e_i : E_i} c_i \rightarrow \multiact_i \sequential P(g_i)$, a boolean expression $\psi$ such that $\freevars(\psi) \subseteq \{d\}$ and a set of indices $J \subseteq I$.
	We define the restricted LPE, denoted by $P^{\psi,J}$, as follows:  
  \begin{align*}
    P^{\psi,J}(d : D) = &\bigplus_{i \in J} \sum_{e_i : E_i} c_i \land \psi \rightarrow \multiact_i \sequential P^{\psi,J}(g_i) \\
    &+ \bigplus_{i \in (I \setminus J)} \sum_{e_i : E_i} c_i \rightarrow \multiact_i \sequential P^{\psi,J}(g_i)
  \end{align*}
\end{definition}

Note that if the boolean expression $\psi$ in Definition~\ref{def:restricted_lpe} is a state invariant for the given LPE then for all closed expressions $\vec{\initval} : \vec{D}$ such that $\interpret{[\vec{d} \gets \vec{\initval}](\psi)}$ holds, it holds that $\interpret{P(\vec{\initval})} \bisim \interpret{P^{\psi,J}(\vec{\initval})}$, for any $J \subseteq I$.
Therefore, we can use a state invariant of an LPE to strengthen all of its condition expressions.

Moreover, a state invariant of the original LPE can \emph{also} be used to restrict the behaviour of the components obtained from a cleave, as formalised in the following theorem.
Note that the set of indices is used to only strengthen the condition expressions of summands that introduce synchronisation,  because the condition expressions of independent summands cannot contain the other parameters as free variables.
Furthermore, the restriction can be applied to independent summands before the decomposition.

\begin{restatable}{theorem}{thminvariant}\label{theorem:invariant}
  Let $P(\vec{d} : \vec{D}) = \bigplus_{i \in I} \sum_{e_i : E_i} c_i \rightarrow \multiact_i \sequential P(\vec{g}_i)$ be an LPE and $(V, \Iind_V, \Isub_V, c^V, \multiact^V, h^V)$ and $(W, \Iind_W, \Isub_W, c^W, \multiact^W, h^W)$ be separations tuples as defined in Definition~\ref{def:derived}.	
	Let $\psi$ be a state invariant of $P$.	
	Given closed expressions $\vec{\initval} : \vec{D}$ such that $\interpret{[\vec{d} \gets \vec{\initval}](\psi)}$ holds the following expression, where $C = \Isub_V \cap \Isub_W$, is a valid decomposition:
 \begin{align*}
   & \hide_{\{\acttag\}} (
    	\allow_{\{\nodata{\multiact_i} \mid i \in I\} \cup \{\nodata{\multiact_i|\acttag} \mid i \in \Iind_V\} \mid i \in \Iind_W\}\}}( \hide_{\{\actsync^i \mid i \in I\}} ( \\
       & \quad \communication_{\{ \actsync^i_V \mid \actsync^i_W \rightarrow \actsync^i \mid i \in I\}} (P^{\psi,C}_V(\subvector{\vec{\initval}}{V}) \parallel P^{\psi,C}_W(\subvector{\vec{\initval}}{W})))
   	  )
     )
 \end{align*}

\end{restatable}

\begin{proof}
  This proof is similar to the proof of Theorem~\ref{theorem:refinedcleave} with the strong bisimilation relation only defined for states where the invariant holds.
  The full proof can be found in Appendix~\ref{appendix:invariant}.
\end{proof}

Observe that the predicate $n \leq 3$ is a state invariant of the LPE $\procmachine$ in Example~\ref{example:delayedmachine}.
Therefore, we can consider the process $\procmachine^{\psi,I}_W$ in Example~\ref{example:naivecleave} for the composition expression, which is finite.
This would yield two finite components, but the state space of $\procmachine^{\psi,I}_W$ is larger than that of $P_W$ in Example~\ref{example:refinedcleave}.

Finally, we remark that the restricted state space contains deadlock states whenever the invariant does not hold.
These deadlocks can be avoided by applying the invariant to the update expression of each parameter instead of the parameter itself without affecting the correctness.

\ML{Indicate scalability: number of parameters/summands?}
\section{Case Study}\label{section:casestudy}

We have implemented an automated translation that, given a partitioning of parameters defined by the user, uses a simple static analysis to obtain components that are guaranteed to satisfy the requirements of a cleave.
Each experiment is a specification written in the high-level language mCRL2~\cite{GrooteM2014}, a process algebra generalising the one of Section~\ref{sec:clpe}.
To apply the decomposition technique we derive an LPE that is behaviourally equivalent using a normalisation procedure that is not discussed in more detail, but which is implemented in the mCRL2 toolset~\cite{BunteGKLNVWWW19}.

We compare size of the original state space to the sizes of the components.
Furthermore, we have minimised the state spaces modulo strong bisimulation.
Finally, we compute the state space of the composition expression applied to these minimised components to determine the effectiveness of the decomposition when compared to the original minimised state space.
We do not present run time and memory usage for these benchmarks.
However, the cost of computing the cleave was in the range of several milliseconds.

\subsection{Alternating Bit Protocol}

The alternating bit protocol (\abp) is a communication protocol that uses a single control bit, which is sent along the message, to implement a reliable communication channel over two unreliable channels~\cite{GrooteM2014}.
The specification contains four processes for the sender, receiver and two unreliable communication channels.

First, we choose the partitioning of the parameters such that one component ($\abp_V$) contains the parameters of the sender and one communication channel, and the other component ($\abp_W$) contains the parameters of the receiver and the other communication channel.
We observe that both components are larger than the original state space, and can also not be minimised further, illustrating that traditional compositional minimisation is, in this case, not particularly useful.

\begin{table}[H]
\centering
\begin{tabular}{ l r r r r}
  \toprule
  Model & \multicolumn{2}{c}{original} & \multicolumn{2}{c}{minimised} \\  \cmidrule(r){2-3} \cmidrule(r){4-5} 
        & \#states & \#trans & \#states & \#trans \\
  \abp   & 182 & 230 & 48 & 58 \\ \hline
  $\abp_V$ & 204 & 512 & 204 & 512 \\ 
  $\abp_W$ & 64 & 196 & 60 & 192 \\
  $\abp^\psi_V$ & 104 & 180 & 53 & 180 \\
  $\abp^\psi_W$ & 58 & 110 & 21 & 108 \\
  $\abp^\psi_V \parallel \abp^\psi_W$ & 172 & 220 & 48 & 58 \\
  $\abp'_V$ & 5 & 35 & 5 & 35 \\
  $\abp'_W$ & 78 & 118 & 28 & 42 \\
  $\abp'_V \parallel \abp'_W$ & 76 & 90 & 48 & 58 \\
  \bottomrule
\end{tabular}
\caption[Table]{Metrics for the alternating bit protocol.}\label{table:sizes}
\end{table}

Further analysis showed that the behaviour of each process heavily depends on the state of the other processes, which results in large components as this information is lost.
We can encode this global information as a state invariant based on the \emph{control flow} parameters.
The second cleave is obtained by obtaining two restricted components ($\abp^\psi_V$ and $\abp^\psi_W$) using this invariant.
This yields a useful decomposition.

Finally, we have obtained a cleave into components $\abp'_V$ and $\abp'_W$ where the partitioning is not based on the original processes.
This yields a very effective cleave as shown in Table~\ref{table:sizes}.

\ML{can this method give us something earlier/other methods cannot?}
\subsection{Practical Examples}

In these experiments we consider several more practical specifications.
We compare the results of the monolithic exploration and the exploration based on the decomposition in Table~\ref{table:hesselink}.
The parameter partitioning for each case is our best effort to obtain the optimal decomposition

\begin{table}[H]
\caption[Table]{State space metrics for various practical specifications.}\label{table:hesselink}

\centering
\begin{tabular}{ l r r r r }
  \toprule
  Model & \multicolumn{2}{c}{exploration} & \multicolumn{2}{c}{minimised} \\ \cmidrule(r){2-3} \cmidrule(r){4-5} 
        & \#states & \#trans & \#states & \#trans \\
  $\hesselink$ & 914\,048 & 1\,885\,824 & 1\,740 & 3\,572 \\ \hline
  $\hesselink_V$ & 464 & 10\,672 & 464 & 10\,672 \\
  $\hesselink_W$ & 97\,280 & 273\,408 & 5\,760 & 16\,832 \\
  $\hesselink_V \parallel \hesselink_W$ & 76\,416 & 157\,952 & 1\,740 & 3\,572 \\ \\
  $\chatbox$ & 65\,536 & 2\,621\,440 & 16 & 144 \\ \hline
  $\chatbox_V$  & 128 & 4\,352 & 128 & 3\,456  \\
  $\chatbox_W$  & 512 & 37\,888 & 8 & 440  \\
  $\chatbox_V \parallel \chatbox_W$ & 1\,024 & 22\,528 & 16 & 144 \\ \\
  $\WMS$ & 155\,034\,776 & \hspace{5mm} 2\,492\,918\,760 & 44\,526\,316 & 698\,524\,456 \\ \hline
  $\WMS_V$ & 212\,992 & 5\,144\,576 & \hspace{5mm} 212\,992 & 2\,801\,664 \\
  $\WMS_W$ & 1\,903\,715 & 121\,945\,196 & 414\,540 & 26\,429\,911 \\
  $\WMS_V \parallel \WMS_W$ & 64\,635\,040 & 1\,031\,080\,812 & 44\,526\,316 & 698\,524\,456 \\
  \bottomrule
\end{tabular}
\end{table}

\noindent The $\chatbox$ specification describes a chat room where four users that can join, leave and send messages~\cite{RomijnS98}.
This specification is interesting because it is described as a monolithic process, which means that compositional minimisation is not applicable.
However, the decomposition technique can be used quite successfully.

The $\hesselink$ specification describes a wait-free handshake register which is presented in~\cite{Hesselink98}.
Finally, we consider the workload management system ($\WMS$) specification described in~\cite{RemenskaWVFTB12}.
For the latter two experiments we found that partitioning the parameters into a set of data parameters and so-called control flow parameters yielded the best results.
Unfortunately, for these practical models we have not managed to improve the results by using an invariant.

We also consider the execution time and maximum amount of memory required to obtain the original state space using exploration and the state space obtained using the presented decomposition technique, for which the results can be found in Table~\ref{table:metrics}.
Here, we consider the maximum amount of memory used, because several tools are run sequentially and only the highest amount used at one time determines the amount of memory required to explore the state space.
Finally, we should note that these times are for the state space obtained under ``exploration'' without considering the final minimisation step.

\begin{table}[H]
\caption[Table]{Execution times and maximum memory usage measurements for various specifications.}\label{table:metrics}

\centering
\begin{tabular}{ l r r r r }
Model & \multicolumn{2}{c}{monolithic} & \multicolumn{2}{c}{decomposition} \\
 \cmidrule(r){2-3} \cmidrule(r){4-5} 
 & time & memory & time & memory \\
 $\chatbox$ & 4.4s & 21.3MB & 0.2s & 14.9MB \\
 $\hesselink$ & 6.9s & 96.9MB & 1.3s & 23.5MB \\
 $\WMS$ & 2.4h & 14.5GB & 0.7h & 11.5GB \\
\bottomrule
\end{tabular}
\end{table}

\section{Conclusion}\label{section:conclusion}

We have presented a decomposition technique, referred to as cleave, that can be applied to any monolithic process with the structure of an LPE and have shown that the result is always a valid decomposition.
Furthermore, we have shown that state invariants can be used to improve the effectiveness of the decomposition.
For practical application we must consider improvements to the heuristics used to obtain the components and especially heuristics to choose the parameter partitioning automatically.
Furthermore, it can also be interesting to consider cleaving into more than two components and even applying it recursively to the resulting components.
Finally, the cleave is currently not well-suited for applying the typically more useful abstraction based on (divergence-preserving) branching bisimulation minimisation~\cite{GlabbeekLT09}.
The reason for this is that $\tau$-actions might be extended with synchronisation actions and tags.
As a result they become visible, effectively reducing branching bisimilarity to strong bisimilarity.

\bibliographystyle{plain}
\bibliography{../bibliography}

\appendix

\section{Proof of Theorem~\ref{theorem:refinedcleave}}\label{appendix:refinedcleave}

The following definition and proposition is due to~\cite{Milner89}.
These are in the context of two LTSs $\lts_1 = (\states_1, s_1, \actions_1, \transitions_1)$ and $\lts_2 = (\states_2, s_2, \actions_2, \transitions_2)$.
We introduce for a binary relation $R \subseteq \states_1 \times \states_2$ the following notation $\bisim R \bisim$ to denote the \emph{relational composition} such that $\bisim R \bisim \,\,= \{(s, t) \in \states_1 \times \states_2 \mid \exists s' \in \states_1, t' \in \states_2 : s \bisim s' \land s' \related t' \land t' \bisim t\}$.

\begin{definition}
  A binary relation $R \subseteq \states_1 \times \states_2$ is a \emph{strong bisimulation up to} $\bisim$ iff for all $s \related t$ it holds that:
  \begin{itemize}
  	\item if $s \transition{\action} s'$ then there is a state $t' \in \states_2$ such that $t \transition{\action} t'$ and $s' \bisim R \bisim t'$.
    
    \item if $t \transition{\action} t'$ then there is a state $s' \in \states_1$ such that $t \transition{\action} t'$ and $t' \bisim R \bisim s'$.
  \end{itemize}
\end{definition}

\begin{proposition}\label{prop:bisimilar}
  If $\related$ is a strong bisimulation up to $\bisim$ then $\related \subseteq\, \bisim$
\end{proposition}

This result establishes that if $R$ is a strong bisimulation up to $\bisim$ then for any pair $(s, t) \in R$ we can conclude that $s \bisim t$.

We introduce two auxiliary lemmas to relate the transition induced by some expression $P \in \simpleprocess$ to the transitions induced by applying the allow, hide and communication operators, in the same order as the composition expression defined in Definition~\ref{def:refinedcleavecomposition}, to $P$.

\begin{lemma}\label{lemma:compositiontransition}
	Given expressions $P, Q \in \simpleprocess$, a set of multi-sets of action labels $A \subseteq 2^{\events \rightarrow \naturalnumbers}$, sets of events $H', H \subseteq \events$, a set of communications $C \subseteq \comms$.
  If $P \transition{\action'} Q$ and $\nodata{\acthide{H'}(\acthide{H}(\actcomm{C}(\action')))} \in A$ then:
	\begin{equation*}
   \hide_{H'} (\allow_{A} (\hide_H (\communication_C (P)))) \transition{\acthide{H'}(\acthide{H}(\actcomm{C}(\action')))} \hide_{H'}(\allow_{A} (\hide_H (\communication_C (Q))))
  \end{equation*}
\end{lemma}

\begin{proof}
	We can derive the following:
	\begin{equation*}
  \inferrule*[Left=Hide]
	{\inferrule*[Left=Allow]
	{\inferrule*[Left=Hide]
	{\inferrule*[Left=Com]
	{P \transition{\action'} Q \\ C \subseteq \comms}
	{\communication_C (P) \transition{\actcomm{C}(\action')} \communication_C (Q) \\ H \subseteq \events}}
	{\hide_H (\communication_C (P)) \transition{\acthide{H}(\actcomm{C}(\action'))} \hide_H (\communication_C (Q))} \\ A \subseteq 2^{\events \rightarrow \naturalnumbers} \\ \nodata{\acthide{H}(\actcomm{C}(\action'))} \in A}
	{\allow_{A} (\hide_H (\communication_C (P))) \transition{\acthide{H}(\actcomm{C}(\action'))} \allow_{A} (\hide_H (\communication_C (Q)))} \\ H' \subseteq \events}
  {\hide_{H'}(\allow_{A} (\hide_H (\communication_C (P)))) \transition{\acthide{H'}(\acthide{H}(\actcomm{C}(\action')))} \hide_{H'}(\allow_{A} (\hide_H (\communication_C (Q))))}
	\end{equation*} 
\end{proof}

\begin{lemma}\label{lemma:othertransition}
	Given expressions $P, Q \in \simpleprocess$, a set of multi-sets of action labels $A \subseteq 2^{\events \rightarrow \naturalnumbers}$, sets of events $H', H \subseteq \events$, a set of communications $C \subseteq \comms$ if:
	\begin{equation*}
		\hide_{H'} (\allow_{A} (\hide_H (\communication_C (P)))) \transition{\action} Q'
	\end{equation*}
  then there are $Q \in \simpleprocess$ and $\action' \in \semactions$ such that $Q' = \allow_{A} (\hide_H (\communication_C (Q)))$, $\action = \acthide{H'}(\acthide{H}(\actcomm{C}(\action')))$, $P \transition{\action'} Q$ and $\nodata{\action} \in A$.
\end{lemma}

\begin{proof}
  Assume that $\hide_{H'} (\allow_{A} (\hide_H (\communication_C (P)))) \transition{\action} Q'$.
  From the structure of the premise conclusion we know that only rule $\textsc{Hide}$ is applicable, which can only be applied for some $Q'' \in \simpleprocess$ such that:
  \begin{equation}
    \inferrule*[Left=Hide]
    {\allow_{A} (\hide_H (\communication_C (P)))) \transition{\acthide{H'}(\action)} Q'' \\ H' \subseteq \events}
    {\hide_{H'} (\allow_{A} (\hide_H (\communication_C (P)))) \transition{\acthide{H'}(\action)} \hide_{H'}(Q'')}
  \end{equation}
  Similarly, we derive the applicability of the $\textsc{Allow}$ and $\textsc{Com}$ rules such that we essentially can obtain (the only possible) derivation shown in the proof of Lemma~\ref{lemma:compositiontransition}.
\end{proof}

\cleavecorrectness*

\begin{proof}
  Pick an arbitrary closed expression $\vec{\initval''} : \vec{D}$.
  Let $(\states_1, s_1, \actions_1, \transitions_1) = \interpret{P(\vec{\initval''})}$ and $(\states_2, s_2, \actions_2, \transitions_2) = \interpret{\hide_{H'} (\allow_A (\hide_H (\communication_C (P_V(\subvector{\vec{\initval''}}{V}) \parallel P_W(\subvector{\vec{\initval''}}{W})))))}$.
  Let $A = \{\nodata{\multiact_i} \mid i \in I\} \cup \{\nodata{\multiact_i|\acttag} \mid i \in \Isub_V\}\}$, $H' = \{\acttag\}$, $H = \{\actsync^i \mid i \in I\}$ and $C = \{ \actsync^i_V | \actsync^i_W \rightarrow \actsync^i \mid i \in I \}$.  
  
  Let $\related$ be the smallest relation $\related$ such that $P(\vec{\initval'}) \related \hide_{H'} (\allow_A ( \hide_H (\communication_C (P_V(\subvector{\vec{\initval'}}{V}) \parallel P_W(\subvector{\vec{\initval'}}{W})))))$, for any closed expression $\vec{\initval'} : \vec{D}$.
  We show that $\related$ is a strong bisimulation relation up to $\bisim$.   
  Pick any arbitrary closed expression $\vec{\initval} : \vec{D}$ and suppose we have the following: $P(\vec{\initval})\,\related\, \hide_{H'} (\allow_A (\hide_H (\communication_C (P_V(\subvector{\vec{\initval}}{V}) \parallel P_W(\subvector{\vec{\initval}}{W})))))$. 

  \begin{itemize}
  \item Case $P(\vec{\initval}) \transition{\action}_1 Q'$.
    There is an index $r \in I$ and a closed expression $l : E_r$ such that for $\subst = [\vec{d} \gets \vec{\initval}, e_r \gets l]$ it holds that $\interpret{\subst(c_r)}$, $\action = \interpret{\subst(\multiact_r)}$ and $Q' = P(\subst(\vec{g_r}))$.
    There are three cases to consider based on the index $r$.
    \begin{itemize}
    \item Case $r \in I \setminus (\Iind_V \cup \Iind_W)$.
      From \REQ{1} this means that $r \in (\Isub_V \cap \Isub_W)$.
			We derive the transitions using requirement \REQ{3}.
      First, from $\interpret{\subst(c^V_r \land c^W_r)}$ it follows that:
      \begin{align*}
        & P_V(\subvector{\vec{\initval}}{V}) \transition{\interpret{\subst(\multiact^V_r|\actsync^V_r(\vec{h}^V_r))}}_2 P_V(\subst(\subvector{\vec{g_r}}{V})) \\ 
        \text{and }& P_W(\subvector{\vec{\initval}}{W}) \transition{\interpret{\subst(\multiact^W_r|\actsync^W_r(\vec{h}^W_r))}}_2 P_W(\subst(\subvector{\vec{g_r}}{W}))
      \end{align*}
      Furthermore, $\interpret{\subst(\multiact_r)} = \interpret{\subst(\multiact^V_r|\multiact^W_r)}$ and by rule $\textsc{Par}$ there is:
      \begin{align*}
        P_V(\subvector{\vec{\initval}}{V}) \parallel P_W(\subvector{\vec{\initval}}{W}) \transition{\interpret{\subst(\multiact_r)|\actsync_V(\subst(\vec{h}^V_r)))|\actsync_W(\subst(\vec{h}^W_r)))}}_2 & \\ P_V(\subst(\subvector{\vec{g_r}}{V}))& \parallel P_W(\subst(\subvector{\vec{g_r}}{W}))
      \end{align*}
      From $\interpret{\subst(\vec{h^V_r} \approx \vec{h^W_r)}}$ it follows that:
      \begin{equation*}
        \acthide{H'}(\acthide{H}(\actcomm{C}(\interpret{\subst(\multiact_r)|\actsync_V(\subst(\vec{h}^V_r)))|\actsync_W(\subst(\vec{h}^W_r))}))) = \interpret{\subst(\multiact_r)}
      \end{equation*}
			From $\nodata{\interpret{\subst(\multiact_r)}} \in A$ we know by Lemma~\ref{lemma:compositiontransition} that:
			\begin{align*}
				\hide_{H'} (\allow_{A} (\hide_H (\communication_C (P_V(\subvector{\vec{\initval}}{V}) \parallel P_W(\subvector{\vec{\initval}}{W}))))) & \transition{\interpret{\subst(\multiact_r)}}_2 \\ \hide_{H'}( \allow_{A} (\hide_H (\communication_C& (P_V(\subst(\subvector{\vec{g_r}}{V})) \parallel P_W(\subst(\subvector{\vec{g_r}}{W}))))))
			\end{align*}
			Finally, $P(\subst(\vec{g_r})) \related \hide_{H'}( \allow_{A} (\hide_H (\communication_C (P_V(\subst(\subvector{\vec{g_r}}{V})) \parallel P_W(\subst(\subvector{\vec{g_r}}{W}))))))$.

    \item Case $r \in \Iind_V$.
      We derive $P_V(\subvector{\vec{\initval}}{V}) \transition{\interpret{\subst(\multiact_r)|\acttag}}_2 P_V(\subst(\subvector{\vec{g_r}}{V}))$, because $\interpret{\subst(c_r)}$ holds.
      There is a transition $P_V(\subvector{\vec{\initval}}{V}) \parallel P_W(\subvector{\vec{\initval}}{W}) \transition{\interpret{\subst(\multiact_r)|\acttag}}_2 P_V(\subst(\subvector{\vec{g_r}}{V})) \parallel P_W(\subvector{\vec{\initval}}{W})$ by From rule $\textsc{ParL}$.
      Furthermore, by definition $\acthide{H'}(\acthide{H}(\actcomm{C}(\subst(\multiact_r) | \acttag))) = \subst(\multiact_r)$ and $\nodata{\subst(\multiact_r)} \in A$.
      From Lemma~\ref{lemma:compositiontransition} we conclude that:
      \begin{align*}
       \hide_{H'} (\allow_{A} (\hide_H (\communication_C (P_V(\subvector{\vec{\initval}}{V}) \parallel P_W(\subvector{\vec{\initval}}{W}))))) \transition{\subst(\multiact_r)}_2 & \\ \hide_{H'}( \allow_{A} (\hide_H (\communication_C (P_V(&\subst(\subvector{\vec{g_r}}{V})) \parallel P_W(\subvector{\vec{\initval}}{W})))))
      \end{align*}
      By \REQ{2} it holds that $\subvector{\vec{g_r}}{W} = \subvector{\vec{\initval}}{W}$.
      Finally, by definition $P(\subst(\vec{g_r})) \related \hide_{H'}( \allow_{A} (\hide_H (\communication_C (P_V(\subst(\subvector{\vec{g_r}}{V})) \parallel P_W(\subst(\subvector{\vec{g_r}}{W}))))))$.

    \item Case $r \in \Iind_W$. Follows from the same observations as $r \in \Iind_V$.
    \end{itemize}

  \item Case $\hide_{H'} (\allow_A ( \hide_H (\communication_C (P_V(\subvector{\vec{\initval}}{V}) \parallel P_W(\subvector{\vec{\initval}}{W}))))) \transition{\action}_2 Q'$. 
 		By Lemma~\ref{lemma:othertransition} there is an expression $Q \in \simpleprocess$ such that $Q' = \hide_{H'} (\allow_{A} (\hide_H (\communication_C (P))))(Q)$, and $\action' \in \semactions$ such that $\action = \acthide{H'}(\acthide{H}(\actcomm{C}(\action')))$, $P_V(\subvector{\vec{\initval}}{V}) \parallel P_W(\subvector{\vec{\initval}}{W}) \transition{\action'}_2 Q$ and $\nodata{\action} \in A$.
    There are three cases where a parallel composition results in a transition.
    Suppose that $P_V(\subvector{\vec{\initval}}{V}) \parallel P_W(\subvector{\vec{\initval}}{W}) \transition{\action'}_2 Q$ is due to:

    \begin{itemize}
      \item Case $P_V(\subvector{\vec{\initval}}{V}) \transition{\action_V}_2 P'_V$ and $P_W(\subvector{\vec{\initval}}{W}) \transition{\action_W}_2 P'_W$ and rule $\textsc{Par}$.
      Then there is a transition $P_V(\subvector{\vec{\initval}}{V}) \parallel P_W(\subvector{\vec{\initval}}{W}) \transition{\action_V + \action_W}_2 P'_V \parallel P'_W$ such that $\action_V + \action_W = \action'$.
      
     	From $\nodata{\acthide{H'}(\acthide{H}(\actcomm{C}(\action_V + \action_W)))} \in A$ we know that $\actcomm{C}(\action_V + \action_W) = \action + \multiset{\actsync_r}$, for some index $r \in I$, because only a single $\acttag$ is allowed with original action labels.
      Therefore, it also follows that $r \in I \setminus (\Iind_V \cup \Iind_W)$.
      Observe that variables (and the related closed expressions) in $\subvector{\vec{d}}{V}$ and $\subvector{\vec{d}}{(W \setminus V)}$ are disjoint.
      There are closed expressions $l, l' : E_r$, $\vec{m} : \subvector{\vec{D}}{(W \setminus V)}$, and $\vec{\initval'}, \vec{\initval''} : \vec{D}$ with the substitutions
      \begin{align*}
        \subst &= [\subvector{\vec{d}}{V} \gets \subvector{\vec{\initval}}{V}, \subvector{\vec{d}}{(W \setminus V)} \gets \subvector{\vec{\initval'}}{V \setminus W}, e_r \gets l] \\
        \text{and }\subst' &= [\subvector{\vec{d}}{W} \gets \subvector{\vec{\initval}}{W}, \subvector{\vec{d}}{(V \setminus W)} \gets \subvector{\vec{\initval''}}{W \setminus V}, e_r \gets l']
      \end{align*}
      such that $\interpret{\subst(c^V_r)}$ holds, $\action_V$ is equal to $ \interpret{\subst(\multiact^V_r)|\actsync^V_r(\vec{h^V_r}))}$ and $P'_V = P_V(\interpret{\subst(\subvector{\vec{g_r}}{V})})$.
      And $\interpret{\subst'(c^W_r)}$ holds, $\action_W = \interpret{\subst'(\multiact^W_r)|\actsync^W_r(\vec{h^W_r})})$ and $P'_W = P_W(\interpret{\subst'(\subvector{\vec{g_r}}{W})})$ and finally $\interpret{\subst(\vec{h^V_r}) \approx \subst'(\vec{h^W_r})}$ holds.
      
      From the requirement \REQ{4} it follows that there is a closed expression $l'' : E_r$ such that for $\rho = [\vec{d} \gets \vec{\initval}, e_r \gets l'']$ that $\interpret{\rho(c_r)}$ holds and $\interpret{\subst(\multiact^V_r) | \subst'(\multiact^W_r)} = \interpret{\rho(\multiact_r)}$.     
      We conclude that $P(\vec{\initval}) \transition{\action}_1 P(\rho(\vec{g_r}))$.
      Furthermore, from $\interpret{\subst(\subvector{\vec{g_r}}{V})} = \interpret{\rho(\subvector{\vec{g_r}}{V})}$, and $\interpret{\subst'(\subvector{\vec{g_r}}{W})} = \interpret{\rho(\subvector{\vec{g_r}}{W})}$ and Lemma~\ref{lemma:bisimulation} it follows that:
      \begin{align*}
        &P_V(\subst(\subvector{\vec{g_r}}{V})) \bisim P_V(\rho(\subvector{\vec{g_r}}{V})) \\
        \text{and }&P_W(\subst'(\subvector{\vec{g_r}}{W})) \bisim P_W(\rho(\subvector{\vec{g_r}}{W}))
      \end{align*}
      By the congruence of strong bisimilarity with respect to $S$ we obtain that:
      \begin{align*}
        \hide_{H'}(\allow_{A} (\hide_H (\communication_C (P_V(\subst(\subvector{\vec{g_r}}{V})) \parallel P_W(\subst'(\subvector{\vec{g_r}}{W}))))) \bisim \\
        \hide_{H'}(\allow_{A} (\hide_H (\communication_C (P_V(\rho(\subvector{\vec{g_r}}{V})) \parallel P_W(\rho(\subvector{\vec{g_r}}{W})))))
      \end{align*}
      Finally, $P(\rho(\vec{g_r})) \related \hide_{H'}(\allow_{A} (\hide_H (\communication_C (P_V(\rho(\subvector{\vec{g_r}}{V})) \parallel P_W(\rho(\subvector{\vec{g_r}}{W}))))$.

    \item Case $P_V(\subvector{\vec{\initval}}{V}) \transition{\action'}_2 P'_V$ and rule $\textsc{ParL}$ such that $P_V(\subvector{\vec{\initval}}{V}) \parallel P_W(\subvector{\vec{\initval}}{W}) \transition{\action'}_2 P'_V \parallel P_W(\subvector{\vec{\initval}}{W})$.
      Pick an arbitrary index $r \in \Isub_V$.
      If $r \in \Isub_V \setminus \Iind_V$ then the action expression contains an action labelled $\actsync^V_r$, which means that $\nodata{\acthide{H'}(\acthide{H}(\actcomm{C}(\action')))} \notin A$. Contradiction.
      
	    As such $r \in \Iind_V$ and the requirement $\freevars(c_r) \cup \freevars(\multiact_r) \cup \freevars(\subvector{\vec{g_r}}{U}) \subseteq \flatten{\subvector{\vec{d}}{U}} \cup \{e_r\}$ holds.
      Therefore, there is a closed expression $l : E_r$ such that for $\subst = [\subvector{\vec{d}}{V} \gets \subvector{\vec{\initval}}{V}, e_r \gets l]$ such that $\interpret{\subst(c_r)}$, $\action' = \interpret{\subst(\multiact_r)|\acttag}$ and $P'_V = P_V(\subst(\subvector{\vec{g_r}}{V}))$.
                        
      We know that for $\subst' = [\vec{d} \gets \vec{\initval}, e_r \gets l]$ that (syntactically) $\subst(c_r) = \subst'(c_r)$ and $\subst(\multiact_r) = \subst'(\multiact_r)$.
      Therefore, $\interpret{\subst'(c_r)}$ holds and $\action' = \interpret{\subst'(\multiact_r)}$.
      Furthermore, from $\subvector{\vec{g_r}}{W} = \subvector{\vec{\initval}}{W}$  (\REQ{2}) we conclude that $P(\vec{\initval}) \transition{\action'}_1 P(\subst'(\vec{g_r}))$.
      Finally, we conclude $P(\subst(\vec{g_r}) \related \hide_{H'}(\allow_{A} (\hide_H (\communication_C (P_V(\subst(\subvector{\vec{g_r}}{V})) \parallel P_W(\subst(\subvector{\vec{g_r}}{W})))))$.

    \item Case $P_W(\subvector{\vec{\initval}}{W}) \transition{\action'}_2 P'_W$ and rule $\textsc{ParR}$, along the same lines as above.
  \end{itemize}
  \end{itemize}
  Using Proposition~\ref{prop:bisimilar} we conclude that $\interpret{P(\subvector{\vec{\initval''}}{V})} \bisim \interpret{ \hide_{H'} (\allow_A (\hide_H (\communication_C (P_V(\subvector{\vec{\initval''}}{V}) \parallel P_W(\subvector{\vec{\initval''}}{W})))))}$
\end{proof}

\section{Proof of Theorem~\ref{theorem:invariant}}\label{appendix:invariant}

\thminvariant*

\begin{proof}   
  Let $\psi$ be a state invariant of $P$ and let $\vec{\initval''} : \vec{D}$ be a closed expression such that $\interpret{[\vec{d} \gets \vec{\initval''}](\psi)}$ holds.
  Let $(\states_1, s_1, \actions_1, \transitions_1) = \interpret{P(\vec{\initval''})}$ and $(\states_2, s_2, \actions_2, \transitions_2) = \interpret{\hide_{H'} (\allow_A (\hide_H (\communication_C (P^{\psi,C}_V(\subvector{\vec{\initval''}}{V}) \parallel P^{\psi,C}_W(\subvector{\vec{\initval''}}{W})))))}$.
    
  Let $\related$ be the smallest relation such that $P(\vec{\initval'}) \related \hide_{H'} (\allow_A ( \hide_H (\communication_C (P^{\psi,C}_V(\subvector{\vec{\initval'}}{V}) \parallel P^{\psi,C}_W(\subvector{\vec{\initval'}}{W})))))$, for any closed expression $\vec{\initval'} : \vec{D}$ such that $\interpret{[\vec{d} \gets \vec{\initval'}](\psi)}$ holds.
  We show that $\related$ is a strong bisimulation relation up to $\bisim$.     
  The rest of the proof follows the same structure as the proof presented in Appendix~\ref{appendix:refinedcleave}.
  
  Consider the case $P(\vec{\initval}) \transition{\action}_1 Q'$.
  First of all, we know that $\interpret{[\vec{d} \gets \vec{\initval}](\psi)}$ holds by definition of $R$.
  In the case $r \in I \setminus (\Iind_V \cup \Iind_W)$, which means that $r \in (\Isub_V \cap \Isub_W)$, we can therefore conclude that $\interpret{\sigma(c^V_r \land c^W_r \land \psi)}$ holds.
  Furthermore, by definition of a state invariant we know that $\interpret{[\vec{d} \gets \sigma(\vec{g_r})](\psi)}$ and thus $P(\subst(\vec{g_r})) \related \hide_{H'}( \allow_{A} (\hide_H (\communication_C (P^{\psi,C}_V(\subst(\subvector{\vec{g_r}}{V})) \parallel P^{\psi,C}_W(\subst(\subvector{\vec{g_r}}{W}))))))$.
  The other case deal with unrestricted summands and as such only the observation that $\interpret{[\vec{d} \gets \sigma(\vec{g_r})](\psi)}$ holds needs to be added.
  
  Consider the case $\hide_{H'} (\allow_A ( \hide_H (\communication_C (P^{\psi,C}_V(\subvector{\vec{\initval}}{V}) \parallel P^{\psi,C}_W(\subvector{\vec{\initval}}{W}))))) \transition{\action}_2 Q'$.
  We can easily see that the restricted condition imply the original condition as well.
  In the first case $P^{\psi,C}_V(\subvector{\vec{\initval}}{V}) \transition{\action_V}_2 P'_V$ and $P^{\psi,C}_W(\subvector{\vec{\initval}}{W}) \transition{\action_W}_2 P'_W$ and rule $\textsc{Par}$ we observe that if $\interpret{\sigma(c^V_r \land \psi)}$ holds then $\interpret{\sigma(c^V_r)}$ holds as well, and similarly if $\interpret{\sigma'(c^W_r \land \psi)}$ holds that $\interpret{\sigma'(c^W_r)}$ holds.
  The remainder of the proof stays exactly the same.
  In the case $P^{\psi,C}_V(\subvector{\vec{\initval}}{V}) \transition{\action'}_2 P'_V$ and rule $\textsc{ParL}$ such that $P^{\psi,C}_V(\subvector{\vec{\initval}}{V}) \parallel P^{\psi,C}_W(\subvector{\vec{\initval}}{W}) \transition{\action'}_2 P'_V \parallel P^{\psi,C}_W(\subvector{\vec{\initval}}{W})$ we observe that $r \in (\Isub_V \setminus \Iind_V)$ and as such the condition expression is not restricted.
\end{proof}

\end{document}